%% file: main.tex
\newif\ifarxiv
\arxivtrue

\ifarxiv
\documentclass[a4paper,UKenglish,cleveref, autoref, thm-restate]{article}
\else 
\documentclass[a4paper,UKenglish,cleveref, autoref, thm-restate]{lipics-v2021}
\fi

\ifarxiv
\usepackage[top=2.7cm]{geometry}
\usepackage{amsmath,amsfonts,amssymb,amsthm}
\usepackage{thm-restate}
\newtheorem{theorem}{Theorem}
\newtheorem{example}[theorem]{Example}
\fi

\usepackage{microtype,xspace,wrapfig,multicol,float} 
\usepackage[font=small,labelfont=bf]{caption}

\usepackage[textsize=tiny,color=lightgray]{todonotes} 
\usepackage{xcolor,soul}
\usepackage{subcaption}
\usepackage{xspace}

\ifarxiv
\usepackage[hidelinks]{hyperref}
\fi

\ifarxiv
\title{Small tile sets that compute while solving mazes\thanks{To appear at the 27th International Conference on DNA Computing and Molecular Programming (DNA27)}}
\makeatletter
\newcommand{\specificthanks}[1]{\@fnsymbol{#1}}
\makeatother
\date{}
\author{Matthew Cook\thanks{Institute of Neuroinformatics, University of Zurich and ETH Zurich} 
\hspace{2em} Tristan St\'{e}rin\textsuperscript{\specificthanks{3}} \hspace{2em} Damien Woods\thanks{Hamilton Institute and Department of Computer Science, Maynooth University, Ireland. Research supported by European Research Council (ERC) under the European Union's Horizon 2020 research and innovation programme (grant agreement No 772766, Active-DNA project), and Science 
Foundation Ireland (SFI) under Grant number 18/ERCS/5746.} 
}

\else 
\title{Small tile sets that compute while solving mazes} 

\titlerunning{} 

\author{Matthew Cook}{Institute of Neuroinformatics, University of Zurich and ETH Zurich}{}{}{}

\author{Tristan Stérin
}{Hamilton Institute, Department of Computer Science, Maynooth University}{tristan.sterin@mu.ie}{https://orcid.org/0000-0002-2649-3718}{Research supported by European Research Council (ERC) under the European Union's Horizon 2020 
research and innovation programme (grant agreement No 772766, Active-DNA project), and Science 
Foundation Ireland (SFI) under Grant number 18/ERCS/5746.}

\author{Damien Woods
}{Hamilton Institute, Department of Computer Science, Maynooth University}{tristan.sterin@mu.ie}{https://orcid.org/0000-0002-0638-2690}{Research supported by European Research Council (ERC) under the European Union's Horizon 2020 
research and innovation programme (grant agreement No 772766, Active-DNA project), and Science 
Foundation Ireland (SFI) under Grant number 18/ERCS/5746.}

\authorrunning{M. Cook, T. Stérin, D. Woods} 

\Copyright{Matthew Cook, Tristan Stérin and Damien Woods} 

\ccsdesc[100]{\textcolor{red}{}} 

\keywords{model of computation, self-assembly, small universal tile set, Boolean circuits, maze-solving} 

\category{} 

\relatedversion{} 

\acknowledgements{}

\nolinenumbers 

\EventEditors{}  
\EventNoEds{2}
\EventLongTitle{27th International Conference on DNA Computing and Molecular Programming (DNA27)}
\EventShortTitle{DNA27}
\EventAcronym{DNA27}
\EventYear{2021}
\EventDate{September 12-17, 2021}
\EventLocation{...}
\EventLogo{}
\SeriesVolume{} 
\ArticleNo{}

\fi

\newcommand{\Nset}{\ensuremath{\mathbb{N}}}
\newcommand{\Rset}{\ensuremath{\mathbb{R}}}
\newcommand{\Zset}{\ensuremath{\mathbb{Z}}}

\newcommand{\calT}{\mathcal{T}}
\newcommand{\simnot}{\mathord{\sim}}

\newcommand{\mawatamfull}{Maze-Walking Tile Assembly Model\xspace} 

\newcommand{\mawatam}{Maze-Walking TAM\xspace} 

\newcommand{\datam}{\mawatam}

\newcommand{\computeGateFourTilesTileLevel}{\ensuremath{6}}

\newcommand{\crossoverGateFourTilesTileLevel}{\ensuremath{34}}

\newcommand{\computeGateSixTilesTileLevel}{\ensuremath{14}}

\newcommand{\crossoverGateSixTilesTileLevel}{\ensuremath{33}}

\newcommand{\nn}{NAND-NXOR\xspace}

\begin{document}

\maketitle

\begin{abstract}
We ask the question of how small a self-assembling set of tiles can be yet have interesting computational behaviour.  We study this question in a model where supporting walls are provided as an input structure for tiles to grow along: we call it the Maze-Walking Tile Assembly Model.  The model has a number of implementation prospects, one being DNA strands that attach to a DNA origami substrate.  Intuitively, the model suggests a separation of signal routing and computation: the input structure (maze) supplies a routing diagram, and the programmer's tile set provides the computational ability.  We ask how simple the computational part can be.

We give two tiny tile sets that are computationally universal in the Maze-Walking Tile Assembly Model.  The first has four tiles and simulates Boolean circuits  by directly implementing NAND, NXOR and NOT gates.  Our second tile set has 6 tiles and is called the Collatz tile set as it produces patterns found in binary/ternary representations of iterations of the Collatz function.  Using computer search we find that the Collatz tile set is expressive enough to encode Boolean circuits using blocks of these patterns.  These two tile sets give two different methods to find simple universal tile sets, and provide motivation for using pre-assembled maze structures as circuit wiring diagrams in molecular self-assembly based computing. 
\end{abstract}

\section{Introduction}
We can think of  solving a maze as performing computation:
the input is a maze, some starting location(s) and an ending location, and the answer to the computation is a yes/no answer signifying whether the exit is reachable from the start, or even an explicit path from start to exit.
Figure~\ref{fig:mazes as circuits}(a,b) shows how a maze encodes a circuit of OR gates:
solving the maze is equivalent to executing the OR circuit with all inputs set to bit 1; and asking about paths in the maze is equivalent to setting some inputs to 1 and seeing which paths have 1 flowing all the way through them. 
It then becomes meaningful to ask about the computational power of systems capable of   solving mazes~\cite{allender2009planar,murphy2012}, for example molecular walker-based systems. 

The difficulty of maze-solving varies with the complexity of the maze, such as 
number of dimensions, 
grid layout versus more general graph, 
degree of nodes, or 
whether graph edges are directed or undirected. 
In computational complexity theory terminology, solving mazes and more general graph reachability problems lie within the class NL~\cite{allender2009planar,moore2011nature,murphy2012}, i.e.~problems solvable on a nondeterministic Turning machine that uses temporary workspace only logarithmic in input length. 
At the simplest level, and perhaps counter-intuitively, a system that solves a directed maze consisting of (a number of possibly disconnected) straight line segments has enough computational power to solve any problem in L, the deterministic version of NL~\cite{Imm1998}.\footnote{The {\sc PathReachability} problem is L-complete:  given a directed graph whose edges form a set of disconnected line segments (in- and out-degree $\leq 1)$, and two nodes $s$ and $t$, is $t$ reachable from $s$? A deterministic Turing machine can start at $s$ and walk along the graph use only logarithmic workspace (in input length) to keep track of the current node, answering ``yes'' if it reaches $t$ and ``no'', if it instead reaches a dead-end. Hence the problem is in L. Conversely, the set of configurations of a deterministic logspace Turing Machine can be encoded as a polynomial-sized instance of {\sc PathReachability} making that problem L-complete~\cite{Imm1998}.}
Thus maze-solving lies between L and NL, depending on the complexity of the setup.

\begin{figure}[t]\centering
    \includegraphics[width=1.0\textwidth]{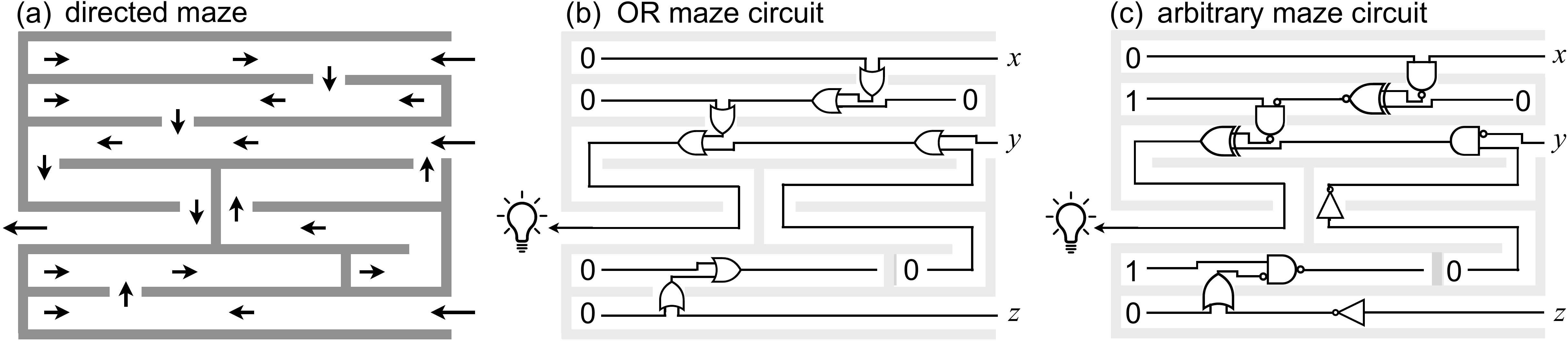}
    \caption{Mazes, computation and Boolean circuits. Solving (a) a {\em directed} maze, where paths have directions, is formally equivalent to executing (b) an OR circuit if we ask: are any of the input bits that are set to 1 connected to the output gate? 
    The example in (b) accepts any 3-bit  input $x,y,z$ that sets $x$ or $y$ to 1, irrespective of~$z$; equivalently the maze is solvable from the top two inputs only.
    Even such a simple setup, allowing for arbitrary mazes, can compute by  solving any suitably-encoded problem in the class nondeterministic logspace~(NL)~\cite{allender2009planar,murphy2012}. 
(c)~We generalise this notion of `computation via maze-solving' in a natural way by having the maze specify arbitrary Boolean gates along the route that need to be evaluated. In the \mawatamfull defined in Section~\ref{sec:def:mawatam}, tiles flow through the maze, building paths from the entrances to the exit, evaluating the circuit as they go.}
    \label{fig:mazes as circuits}
\end{figure}

Here, we suggest two modifications to the maze-solving problem, which are expressive enough to endow maze solvers with significant computational power (their prediction problem becomes P-complete),
yet, we contend, simple enough to be experimentally feasible using DNA engineering and computing principles. 
The first, and most important, is that we generalise mazes to have paths patterned with logic gates that must be solved in order to pass by them (Figure~\ref{fig:mazes as circuits}(c)). For a maze-walker this would mean it should be able to input one or two bits of information from the site it stands upon, compute, and then output one or two bits to adjacent sites. 
The second, mainly to keep things simple, is that we assume mazes are directed (meaning a pair of adjacent positions have one directed edge between them that dictates the direction of information flow) and have no cycles. 
Since we allow for fanout of 0, 1 or 2 per site, one needs to generalise the typical notion of maze-solving somewhat: 
Are walkers replicating themselves to handle fanout of~2? Or are they leaving little bit-encoding messages for other walkers/themselves to pick up later? How do they handle fanin of~2?
These considerations lend themselves to various models, however here we focus on  having information-manipulating {\em tiles} flow through the maze, much like lava flowing down a complex volcanic hillside, but clever lava that computes as it moves.  Our model is called  the \mawatamfull, or \mawatam.

The programmer specifies a set of square tiles, with glues on the sides. 
A problem instance, or maze,  is a set of polyominos, painted with information-encoding glues. 
Starting at special input locations, tiles attach one at a time, asynchronously and in parallel,  wherever they match glues on two sides.\footnote{The model is equivalent to the abstract Tile Assembly Model~\cite{rothemund2000program,Winf98,PatitzSurveyJournal,DotCACM}, with multiple disconnected seed assemblies, and where we have all tile bindings are by attachment to an assembly by two matching glues.}
A typical maze can be thought of as sending a unary (``route finding'') signal, whereas our mazes send bits and allow them to meet,  interact and be changed.

In this setting, if we allow arbitrary numbers of tiles (or a clever enough walker, or a complex enough asynchronous cellular automaton rule)  it is not difficult to see how to simulate arbitrary Boolean circuits. 
Take a circuit, make it planar by replacing each wire crossing with a crossover gate, then lay the circuit out on a maze-like grid with input gates on the east, and the output gate on the west.
Then simply build a maze with walls tracing out the circuit wiring diagram and painted with arrows (wire directions) and logic gates, and require the output bit(s) to satisfy the circuit logic. 
The question we ask is:  How clever does the maze-solver need to be in this computational setting? More precisely, we ask how many tile types are needed to execute any Boolean circuit in the \mawatam?

\subsection{Main results}

\begin{figure}
    \centering
    \includegraphics[width=\textwidth]{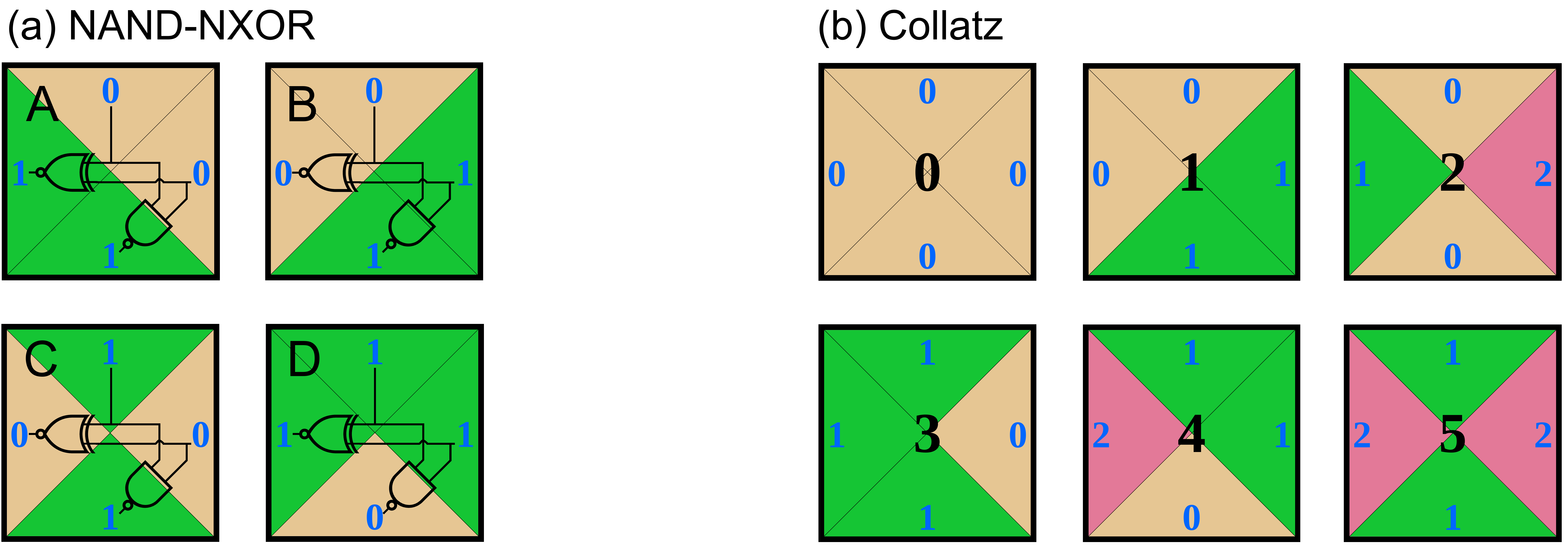} 
    \caption{Two small tiles sets.
        (a) \nn tile set with 4 tile types. 
        The south side computes the NAND of north and east,
        and west computes the NXOR of north and east.
        (b) Collatz tile set with 6 tiles, named for its relationship to the Collatz problem.
    }\label{fig:tilesets}
\end{figure}

Our first main result is for the \nn tile set shown in Figure~\ref{fig:tilesets}(a). In 
the theorem statement, by {\em simulated} we mean that the function computed by the circuit $c$ is also computed by an instance of the \mawatam (see Section~\ref{sec:def:mawatam}).

\begin{restatable}[]{theorem}{thmNANDNXOR} 
  \label{thm:NANDNXOR}
  \normalfont
  Any Boolean circuit $c$ is simulated by the 4-tile \nn tile set in the \mawatam using assemblies
  containing $\leq \computeGateFourTilesTileLevel$ tiles per gate and 
  $\crossoverGateFourTilesTileLevel$ tiles per crossover gate in a planarisation of $c$.
\end{restatable}

Our second main result is for the Collatz tile set which has 6 tiles (Figure~\ref{fig:tilesets}(b)) and is so-named because of its ability to embed iterations of the Collatz function (see Appendix~\ref{app:iwantmoreCollatz}). 

\begin{restatable}[]{theorem}{thmCollatz} 
  \label{thm:Collatz}\label{thm:six tiles}
  \normalfont
    Any Boolean circuit $c$ is simulated by the 6-tile Collatz tile set in the \mawatam using assemblies
  containing $\leq \computeGateSixTilesTileLevel$ tiles per gate and 
  $\crossoverGateSixTilesTileLevel$ tiles per crossover gate in a planarisation of $c$.
\end{restatable}

We finish this section with a discussion of our two tile sets and some future directions. 
Section~\ref{sec:related} sets these results in the context of other theoretical results 
 and experimental directions. 
 Section~\ref{sec:defs} defines the \mawatam.  We prove our two main theorems in Sections~\ref{sec:four tiles} and~\ref{sec:six tiles}.
 Appendix~\ref{app:iwantmoreCollatz} gives some background on the Collatz tile set.

\subsection{Discussion:  the \nn and Collatz tile sets}
Theorems \ref{thm:NANDNXOR} and \ref{thm:Collatz} place focus on the size of assemblies that simulate gates. 
They omit estimates of the additional tiles (assemblies) required for the circuit wiring diagram, which warrants comment. 
Our work is partially motivated by a desire to build instances of the \mawatam, and in doing so we would  highly optimise any implemented circuit wiring diagram. 
Example circuit implementations,  that recognise 3-bit prime numbers, are shown in Figures~\ref{fig:gadgets4}(j3) and \ref{fig:gadgets6}(j1), both of which are optimised for short wire length. 
If we want to have a general wiring procedure for all circuits, and thus not optimised for particular classes of circuits,  the overhead incurred will be rather large, typically $O(s^2)$ space for a circuit with $s$ gates~\cite{graphOnAGrid}. 
In practice we would not use such overly-bloated constructions.

The \nn tile set was found by explicitly trying to find a small tile set: hence its use of a universal gate (NAND) on the south side (output). The NXOR gate (west side) helps with wire routing allows for even smaller gates than going via NAND-only-based circuit simulation. 
The Collatz tile set came out of thinking about iterations of the Collatz function in a local digit-by-digit, or tile-by-tile, way. 
In~\cite{Collatz2} a cellular automaton-like model is shown to simulate instances of the Collatz function---assemblies of our Collatz tile set show up in iterations (configurations) of that model. 
The Collatz tile set, along with the non-local rule in~\cite{Collatz2} (which can be simulated by the addition of two additional tile types, see Appendix~\ref{app:iwantmoreCollatz}), is expressive enough to run Collatz. 
Here we applied computer search to the Collatz tile set to search for seed structures and assemblies that could be used to compute more generally. 
We leave as an open question as to what extent such structures, or other computational structures, naturally appear during iterations of the Collatz function---something the Collatz tile set might help us see. 

For running Boolean circuits, if the only metric we cared about was tile set size, the \nn tile set wins. 
However, looking beyond circuits, 
the Collatz tile set is capable of directly implementing certain arithmetical operations, such as computing powers of 2, powers of 3, and converting from base 3 to base 2~\cite{Collatz2} (see Appendix~\ref{app:iwantmoreCollatz}). 
These constructions use much simpler connected seeds than those given in the proof of Theorem~\ref{thm:Collatz}, and  lead to more efficient (smaller) assemblies than computing via tiles-simulating-circuits, for these kinds of arithmetical problems.   
In this paper, we used computer search to find that tiles capable of such arithmetical operations are also capable of running circuits, we leave it as future work to discover what other operations they are efficiently capable of. 

Theorems 1 and 2 prove that the problem of predicting a tile at distance $n$ from a size $n$ connected seed, is P-hard (and in fact it is also P-complete if we assume directed/deterministic growth~\cite{SolWin07} since a deterministic Turing Machine  simulates the entire assembly process in time polynomial in $n$). 
It is natural to ask if having maze-like (i.e.~disconnected) seeds is necessary for such computational efficiency: we conjecture ``yes''. That is,     
for both tile sets, we conjecture that prediction of the tile type that goes at a given position, at distance $n$ from a size $n$ connected seed and assuming directed growth, is in the complexity class NL.  
In particular this would mean that simulation of arbitrary Boolean circuits in the direct manner shown here is impossible, assuming the widely-believed conjecture NL $\neq$ P. 
For the Collatz tile set, and for connected seeds of a certain form, we know that prediction is in NL (Appendix~\ref{app:iwantmoreCollatz}). 
If one could show that prediction is P-hard, for seeds/inputs that represent natural numbers that occur during iterations of the Collatz function, one could in fact show that the Collatz process embeds rather powerful computational capabilities. 
Certainly a result of that form would change the perspective on the Collatz conjecture itself. 

Our results were developed with assistance of a simulator: \url{https://github.com/tcosmo/mawatam}. The reader is invited to experience the results of this paper through the simulator.

\subsection{Future work}
Experimentally, future work involves implementing instances of the \mawatam in the wet-lab, for instance, using a DNA origami as the underlying structure to encode maze seeds \cite{Chao2019}, building on the systems discussed in Section~\ref{sec:related experiment}.
One experimentally-relevant criticism of this work could be to ask why we focus on such small tile sets when we know that with DNA it is possible to build systems with hundreds of algorithmic DNA tiles~\cite{IBCs}.  
First, we would say that no algorithmic system of such a high tile complexity, and that runs on the back of a DNA origami, has been engineered to date. 
Secondly, and of more relevance to this work, is that we are exploring the fundamental boundary and complexity trade-offs between computational power and systems size.

Theoretically, our work leaves open the following questions:
\begin{itemize}
\item Can Boolean circuit simulation, or any kind of universal computation, be achieved in the \mawatam using tile sets with  less than $4$ tiles?
\item Can interesting behaviour occur in the \mawatam with just 1 tile? (At first sight, this question may look odd, however one could imagine encoding a bit  by the absence or presence of a tile at a given position in the final assembly, leaving room for expressiveness in the \mawatam with 1 tile.)
\item Is the \mawatam, with $\leq 4$ tiles, intrinsically universal~\cite{IUSA,IUsurvey} for the aTAM?
\end{itemize}

\ifarxiv
\newpage
\fi
\section{Related work: theoretical and experimental}\label{sec:related}

\subsection{Other routes to finding small universal tile sets}\label{sec:related theory}
Existing small/simple universal models of computation~\cite{WoodsNearySurvey} include 
the efficiently universal~\cite{Cook,neary2006p} 2-state one-dimensional cellular automaton Rule 110,
as well as universal Turing machines with just 22 instructions (5 states \& 5 symbols, or 4 states \& 6 symbols)~\cite{neary2009four, Rogozhin96}  or even just with 8 instructions (3 states, 3 symbols, but with the tape input embedded in an infinitely repeated pattern)~\cite{NearyWoodsWeakly}. 

In the context of the theory of molecular computing, and algorithmic self-assembly in particular, the smallest computationally universal self-assembling tile set to date seems to be a~7-tile system that can be derived from~\cite{IBCs}.\footnote{In Figure~S4(b), SI~A,~\cite{IBCs}, gates $g$ and $f$ can be used to simulate Rule 110, and that in turn can be simulated by 4 tiles each. These 8 tiles can be further optimised to 7 tiles by sharing one glue type between both half-layers.}  
However, that construction leads to large spatial blowup via Rule 110 simulation of  $O(s^4 \log^2 s)$ for circuits of size $s$  (Corollary~S1.3, SI-A~\cite{IBCs}). 
Another construction uses $O(w^2 d)$ tile types (for a depth $d$, width $w$ circuit), essentially by hardcoding the routing of the circuit diagram in tile types (Theorem S1.5, SI-A \cite{IBCs}). 
Even direct implementation of a small universal Turing machine as a self-assembling tile set, using known methods, although presumably achievable with a few dozen tile types, would require large input encodings~\cite{WoodsNearySurvey}.
Other methods to obtain a single universal, or intrinsically universal, tile set, or even a single tile, also use indirect and large, albeit constant-factor in some cases, encoding methods~\cite{IUSA,2HAMIU,OneTile,SoloveichikWinfree}. 

By allowing for more tile types than our constructions, one could have a maze with glues that explicitly encode gate type (one of sixteen), as well as glues encoding two bits at a time: that way a single tile attachment event could read two bits and a gate type simultaneously. This idea yields a constant-size tile set with perhaps a few dozen tile types. 
Although larger than ours, such an approach would have experimental merit. 
Cantu, Luchsinger,  Schweller, and Wylie simulate Boolean circuits with tiles in a covert manner~\cite{cantu2021covert}.

\subsection{DNA-based implementations and related models}\label{sec:related experiment}
 As future work we plan to give DNA-based designs and  implementation for the \mawatam. 
We imagine a 2D information-encoding structure that provides the maze pattern, for example a single flat DNA origami~\cite{rothemund2006folding}, or several DNA origamis tiled together~\cite{woo2011programmable,liu2011crystalline,tikhomirov2017fractal,tikhomirov2017programmable}, or perhaps even a suitable DNA DX-tile, or single-stranded tile,  structure~\cite{SST2D,erikNed1998,Yan8103,IBCs}.
DNA-based systems for maze-solving have been implemented experimentally:  
using DNA origami (for the maze) along with hairpin activation \cite{Chao2019} or controlled opening of track locations~\cite{Wickham2012} for movement.
The phenomenon of DNA condensation was also used for maze exploration~\cite{pardatscher2016dna}. 
Computation via tile-attachment in the \mawatam could be implemented using design principles from algorithmic DNA self-assembly~\cite{IBCs,evans2014crystals},
DNA-based molecular walkers that walk on 1D tracks and 2D DNA origami surfaces~\cite{walker_yin2004, walker_sherman2004, walker_Sahu2008, walker_omabegho2009, gu2010proximity, walker_cargoSorting2017},
and other DNA systems that compute on  surfaces~\cite{localised_dna_hairpin_chain_Bui, localised_dna_origami_circuit,
    dna_based_circuit_cancer_membrane_Song2019, pmid28977499, Chatterjee2017}.  
Finally, there has been some theoretical and simulation-based analyses of molecular walkers~\cite{Dannenberg, walker_computing_REIF20091428, lakin2014abstract,dalchau2015probabilistic} including maze-solving walkers~\cite{darkoMazeWalker},
as well as papers that study computation on surfaces~\cite{surface_qian2014, surface_clamons2020, surface_swap_Brailovskaya_2019} using a similar setup to ours but without molecular orientation and using different rule formats. 
All of these models (and ours) describe sub-classes of asynchronous cellular automata.

\section{Definitions}\label{sec:defs}
\subsection{\mawatam definition}\label{sec:def:mawatam}
A {\em maze} is collection of non-intersecting polyominos positioned on $\Zset^2$  where each exterior unit-length square-side polyomino edge  is labeled with a glue $g = (g',p)$ where $g' \in G$ is from a finite set of {\em glue types} $G$,  that includes the null glue, and  $p \in \{ z + 0.5 \mid z\in\Zset  \}^2$ is a glue position.
An {\em instance} of the \mawatam $\calT = (T, M)$ has a set of tile types $T$, where each $t \in T$ is a unit-sized square whose four sides  labelled with four glue types from $G$, and a maze $M$.
The process of {\em self-assembly} proceeds by tiles (instances of tile types) attaching asynchronously, and one at a time, wherever they match non-null glues on two sides (i.e.~two-sided cooperative binding in the abstract Tile Assembly Model~\cite{Winf98,rothemund2000program,PatitzSurveyJournal,DotCACM}).
An {\em assembly} is a maze with tiles attached (thus, assemblies may be connected or disconnected in 2D), and a {\em terminal assembly} is an assembly such that no tile can be attached. 

The tile set $T$ is said to {\em compute the function} $f : \{0,1\}^n \rightarrow \{0,1\}$ in the \mawatam if there is a maze $M'$ with $n$ empty (no tile) tile positions $p_0,p_1, \ldots, p_{n-1} \in \Zset^2$  
and an empty (no glue) output glue position~$o \in \{ z + 0.5 \mid z\in\Zset  \}^2$,
such that adding $n$ {\em input tiles} at $p_0,p_1, \ldots, p_{n-1}$ to $M'$ is the new maze called $M_x$ where the process of self-assembly on $M_x$ yields a set of terminal assemblies that each have the bit $f(x)$ encoded by the glue at position $o$. 
(Here, we imagine a many-one encoding function from glue types to bits.) 
 
\mawatam systems may be {\em directed} (one terminal assembly), or {\em undirected} (several terminal assemblies). 
In this paper the systems we study are directed, which is equivalent to saying that, for all sequences of tile additions, at each position $p\in\Zset^2$, there is at most one choice for what tile appears at $p$. 
Thus, in this paper, for a function $f$, for each $x\in\{0,1\}^n$ there is an associated maze $M_x$ such that  $\calT_x = (T,M_x)$ has  a single terminal assembly that is said to compute $f(x)$. Finally, a Boolean circuit $c$ (defined below) is said to be simulated by a tile set if the tile set computes the same function as $c$.

\subsection{Boolean circuit definition}\label{sec:defs:circuits} A Boolean circuit is a
directed acyclic graph, where edges are called {\em wires}, and nodes are called {\em gates}
and are labelled. In this paper, gates have out-degree 1 or 2, except for output
gates that have out-degree 0. Also, a node's label is one of: input (with
in-degree~0), output (with in-degree 1, out-degree~0), constant~0 or constant~1
(in-degree~0), fanout gates (in-degree~1, out-degree~2; makes two copies of its
input), or is one of the {\em compute} gates ($\neg$, NOT of in- and out-degree
$1$, or any of the in-degree 2 out-degree 1 gates that compute functions on
bits, e.g. OR,  AND, NAND, NXOR,\footnote{In this paper we use the notation
NXOR($x,y$) = NOT(XOR($x,y$)) (and read ``NOT exclusive OR'') to denote what is
more commonly, but confusingly, written XNOR (read ``exclusive NOR'').} etc.).
Also, we define an additional gate called a {\em crossover gate}
(in- and out-degree of 2) which swaps its inputs, used to planarise a non-planar
circuit (see below). Circuits compute, from the  input gates and constant gates to the
output gate, by modifying bits according to the functions specified by gate labels.

The \emph{size} of a circuit is its number of gates, and its \emph{depth} is the length of the longest path from any input gate to the output gate.
A circuit $c$ computes a Boolean (no/yes) function $f :\{ 0,1\}^n \rightarrow {0,1} $ on  $n\in\Nset$ Boolean variables, by its gates computing the bit value at the output in the usual way from the $n$ input bits.
A circuit is said to be planar if its graph is planar (can be laid out in the plane without wire crossings).

A {\em planarisation} of a Boolean circuit $c$ is another Boolean circuit $\hat c$ where $\hat c$ computes the same function as $c$, has a planar embedding in $\Rset^2$, and
$\hat c$ has exactly the gates of $c$ plus zero or more 2-in 2-out crossover gates (that allow crossing of signals between a pair of wires that would otherwise intersect in the plane). 
In other words, $c$ is converted to $\hat c$ by adding crossover gates so that $\hat c$ has a planar embedding.
An example is shown in Figure~\ref{fig:gadgets4}(j2). 
A {\em planar Boolean circuit} $c$ is a Boolean circuit where $\hat c = c$, i.e. $\hat c$ has zero crossover gates. 

\section{Four tiles: the \nn tile set}\label{sec:four tiles}
\begin{figure}
    \includegraphics[width=\textwidth]{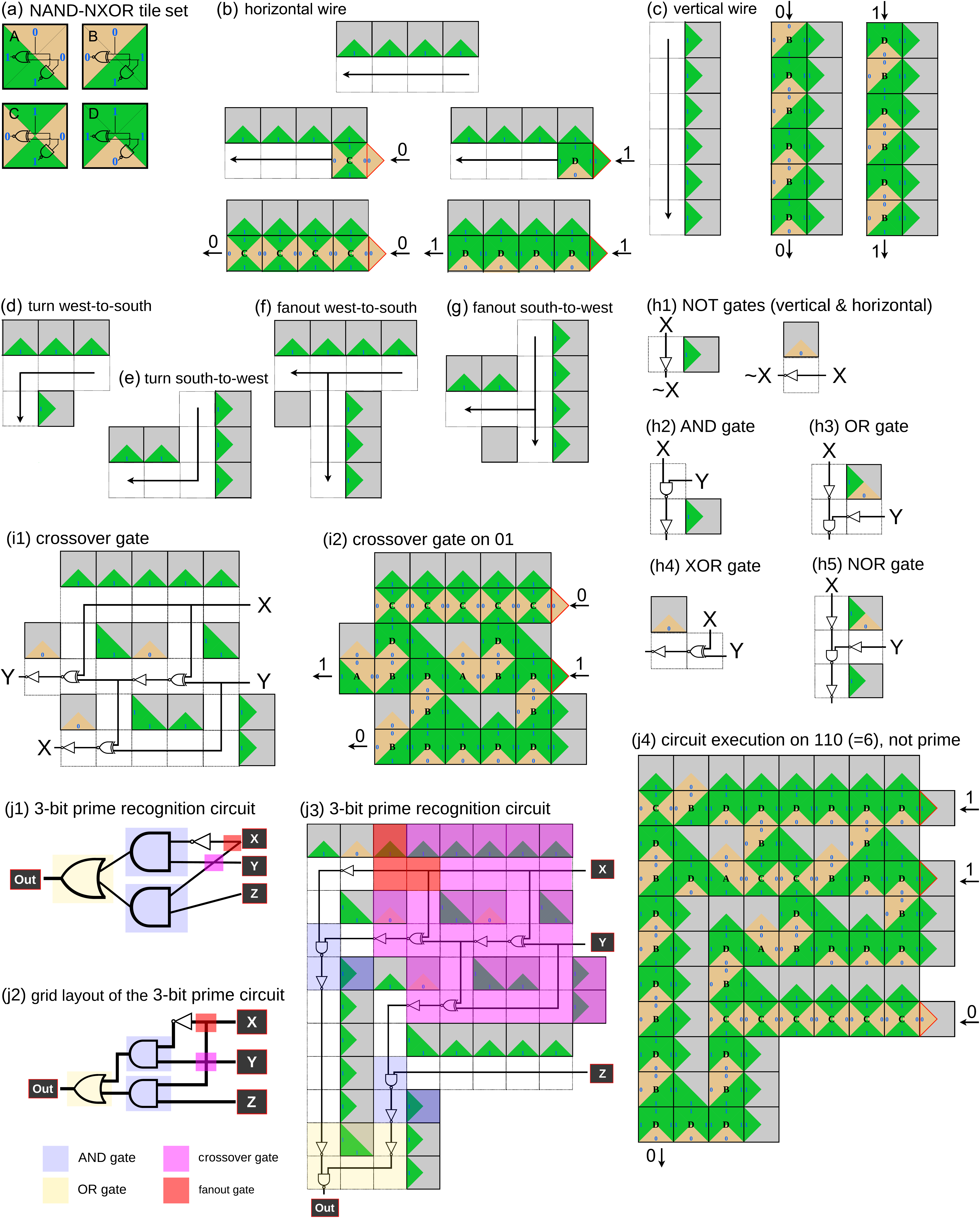}
    \vspace{-4ex}\caption{Circuit-simulating gadgets for the \nn  tile set.
            In all parts of the construction growth proceeds to the west and south (and never north nor east).
    (a) \nn tile set. 
      Seed structures to implement (b) horizontal west-growing and (c) vertical south-growing wires. Examples of communicating of 0 and 1 are shown for each. 
      Vertical wires are of even length; in cases where odd length is required we use a horizontal NOT gates during a turn from south-to-west (see proof of Theorem~\ref{thm:NANDNXOR}).  
        (d) Turn west-to-south, (e) turn south-to-west, 
        (f) fanout west-to-south, and 
        (g) fanout south-to-west.
        The two isolated unit-size squares in (f,g) are there only to prevent unintended cooperative growth after a fanout.
        (h1--5) Various logic gates (full set in Figure~\ref{fig:all_gates4}). 
       (i1) Crossover gate with an example in (i2) with design based on the 3 XOR gates construction given in \cite{cantu2021covert}. 
        (j1)~An example Boolean circuit that decides whether a 3-bit number is prime. 
        (j2) Circuit converted to a grid layout and (j3) implemented using \nn tile gadgets. 
        The implementation in (j3) is somewhat optimised for space efficiency. 
        (j4) The terminal assembly (execution) for the circuit example on non-prime input $6_{10} = 110_2$. 
    }\label{fig:gadgets4}
\end{figure}

The \nn tile set is depicted in Figure~\ref{fig:gadgets4}(a). 
One of the ideas underlying all of the constructions in this paper can be understood  by the way horizontal wires are implemented with the \nn tile set, Figure~\ref{fig:gadgets4}(b). 
A specific $n \times 1$ polyomino seed advertises ``1'' glues along its south side, which facilitates propagation to the west of any bit presented as a glue coming from the east, following the assembly rules prescribed by the tile set. 
As described in the proof of Theorem~\ref{thm:NANDNXOR},
the implementation of Boolean circuits using the \nn tile set is based on canonical constructions of logic gates exploiting NAND, NOT and NXOR functions as primitive building blocks, Figure~\ref{fig:gadgets4}(h1-h5).

\thmNANDNXOR* 
\begin{proof}
    A circuit is simulated by appropriately placing gadgets together to form a maze. 
     
    {\bf Tiles simulating wires and gates.}
    We will show that the gadgets in Figure~\ref{fig:gadgets4} are building blocks (for a maze) that advertise glues designed to force directed growth when given appropriate bit-encoding glue input(s). 
    
    Figure~\ref{fig:gadgets4}(b,c) details how the \nn tile set simulates horizontal and vertical wires. 
    Vertical tile-wires have a parity constraint: in a vertical wire carrying the bit $x\in\{0,1\}$, 
    every second tile correctly advertises $x$ to the south, and every other tile advertises its negation $\simnot x$. 
    If the circuit's layout requires a turn from south-to-west, from an {\em odd length} vertical wire (advertises $\simnot x$) 
    then a single horizontal negation gadget (Figure~\ref{fig:gadgets4}(h1, right)) is placed at the bottom of the wire to change the signal to $x$ (correct the ``error''). 
    With that correction, vertical and horizontal wire segments can be used to send a signal  from the origin to any location in the south-west quadrant of $\Zset^2$.

Figure~\ref{fig:gadgets4}(d--g,h1--h5,i1) shows  
    two turns (south-to-west and west-to-south) and  two kinds of fanout-2 gates, 
    as well as a number of compute gates and a crossover gate.
    In addition NAND, and NXOR, gates are shown in Figure~\ref{fig:gadgets4}(a): present inputs {\sf x}, {\sf y} at North and East, 
    and read NXOR({\sf x}, {\sf y}) on West  and/or  NAND({\sf x}, {\sf y}) on South.  
    (For completeness, Figure~\ref{fig:all_gates4} gives direct simulations of all 16 possible gates with 1 or 2 inputs and one output.) 
    No gate is larger than NOR (see Figure~\ref{fig:gadgets4}(h5) and Figure~\ref{fig:all_gates4}), which uses $\computeGateFourTilesTileLevel$ tiles.
    The crossover gate is simulated using $\crossoverGateFourTilesTileLevel$ tiles (intuitively, it uses a well-known idea of implementing crossover with three XOR gates and three fanout gates).
  This gives the size bounds on tiles per gate and crossovers in the theorem statement.
  
  We claim that each gadget in Figures~\ref{fig:gadgets4}(b--g,h1--h5,i1) and Figure~\ref{fig:all_gates4} is directed, meaning that after input glue(s) are given to the gadget, then for each unit-sized outlined/dotted empty square region in the gadget there is exactly one tile type that can be placed. 
  This can be seen by noting that (i) for all gadgets, and all inputs to a gadget, tiles attach using their North and East sides only, and by (ii) the fact that the \nn tile set is deterministic on North and East sides.

   {\bf Laying the circuit out on a grid.}
   For the  Boolean circuit $c$, let $\hat c$ be its planarisation as defined in Section~\ref{sec:defs:circuits}; a planarisation always exists---just draw the circuit on the plane replacing each of the $s'\in \Nset$ wire crossings with a crossover gate (various planarisations may be used to optimise $s'$, or other circuit parameters).

    Second, we layer $c$: meaning that we organise gates (including crossover gates) of $c$ into consecutive layers with layer 0 containing all input and constant gates, and so that layer~$i$ contains gates that take their inputs from the outputs of gates in layers $< i$. The number of layers is equal to the depth $d$ of $c$, with the output gate being the sole gate in layer $d-1$.
    More precisely, layer $i$ is located at x-coordinate $- i$ (our convention is to draw circuits from right to left).

    Third, we increase the height between gates, and width between layers, so that there is enough room to draw all wires so that they are composed of horizontal and vertical segments only (where information flows to the west and to the south, respectively), that meet at right angles (thus wires have south-to-west and west-to-south turns, only).  
   We call the resulting circuit a grid-layout circuit, and an example   given in Figure~\ref{fig:gadgets4}(j2).    
    Using the gadgets described above, the maze/seed structure traces out the wires and gate locations according to the south-west grid-layout circuit, leaving enough room so that gates and wires to not intersect. 
    
    {\bf Computation.}
  For any circuit $c$ we have described (at a high level) how to lay out a maze $M'$,  in the notation of Section~\ref{sec:def:mawatam}. We next need to encode circuit inputs, as follows. 
    Since input gates are instances of gates, we assume that in $M'$ there are $n$ tile positions that are empty and positioned adjacent to wires (so that their bit values will feed into a layer of gates via horizontal wire gadgets).
   Let $n$ be the number of inputs to $c$ and let $x = x_0 x_1 \cdots x_{n-1} \in \{0,1\}^n$ denote an input to~$c$. 
   To the maze $M'$ we add $n$ more tiles so that the $n$ input glue positions of the maze are of respective types $x_0 x_1 \cdots x_{n-1}$, to give an maze $M_x$ that encodes $x$  (the example in Figure~\ref{fig:gadgets4}(j4) has 3 encoded input bits).

    Assembly proceeds, starting at each of the $n$  input glues in parallel (and at any positions that encode 0/1 constant bits), according to the \mawatam definition (Section~\ref{sec:def:mawatam}). 
    Throughout the entire self-assembly process, at each position there is exactly one tile type that can be placed (this is because it is true for individual gadgets as already argued). 
    Also, the self-assembly process terminates, for the simple reason that no tile can attach outside of the bounding box of the maze $M_x$.  
    Thus one terminal assembly is eventually produced, that by its definition, encodes an execution of the circuit $c$ with the output bit presented at the glue position that represents the simulated circuit output gate (labeled ``out'' in the example in Figure~\ref{fig:gadgets4}(j3)). 
\end{proof}

\begin{example}\label{ex:circuit4}\normalfont
Figure~\ref{fig:gadgets4}(j1-j4) illustrates the general construction described in Theorem~\ref{thm:NANDNXOR} in the context of a circuit that recognises prime numbers on 3 bits, i.e. the circuit will output $1$ if and only if $xyz \in \{010, 011, 101, 111\}$ which are the binary encodings of numbers $\{2,3,5,7\}$. The circuit implements the formula: (((NOT $x$) AND $y$) OR ($x$ AND $z$)) and uses one crossover as well as one fanout gate, Figure~\ref{fig:gadgets4}(j1). To facilitate the final \mawatam implementation, the circuit is laid out on a grid using only south-to-west and west-to-south turns, Figure~\ref{fig:gadgets4}(j2). Then, the circuit is implemented with tiles, Figure~\ref{fig:gadgets4}(j2), using the gadgets of Figure~\ref{fig:gadgets4} and finally, the circuit executes on input $110_2 = 6$ and outputs $0$ as $6$ is not prime, Figure~\ref{fig:gadgets4}(j4). Note two details: 
(1)~The implementation of the crossover gate,  Figure~\ref{fig:gadgets4}(i1), contains three embedded XOR gadgets and three embedded fanout gadgets---using tiles to implement a known construction to simulate crossover with XORs. 
(2)~The way the OR gate is implemented in Figure~\ref{fig:gadgets4}(j3) (yellow overlay) is slightly different than Figure~\ref{fig:gadgets4}(h3) as the negation of the east-coming input is performed vertically instead of horizontally; this is an optimisation that exploits the difference in length parity of the two vertical wires coming in to the gate. 
\end{example}

\begin{figure}[h]
    \includegraphics[width=\textwidth]{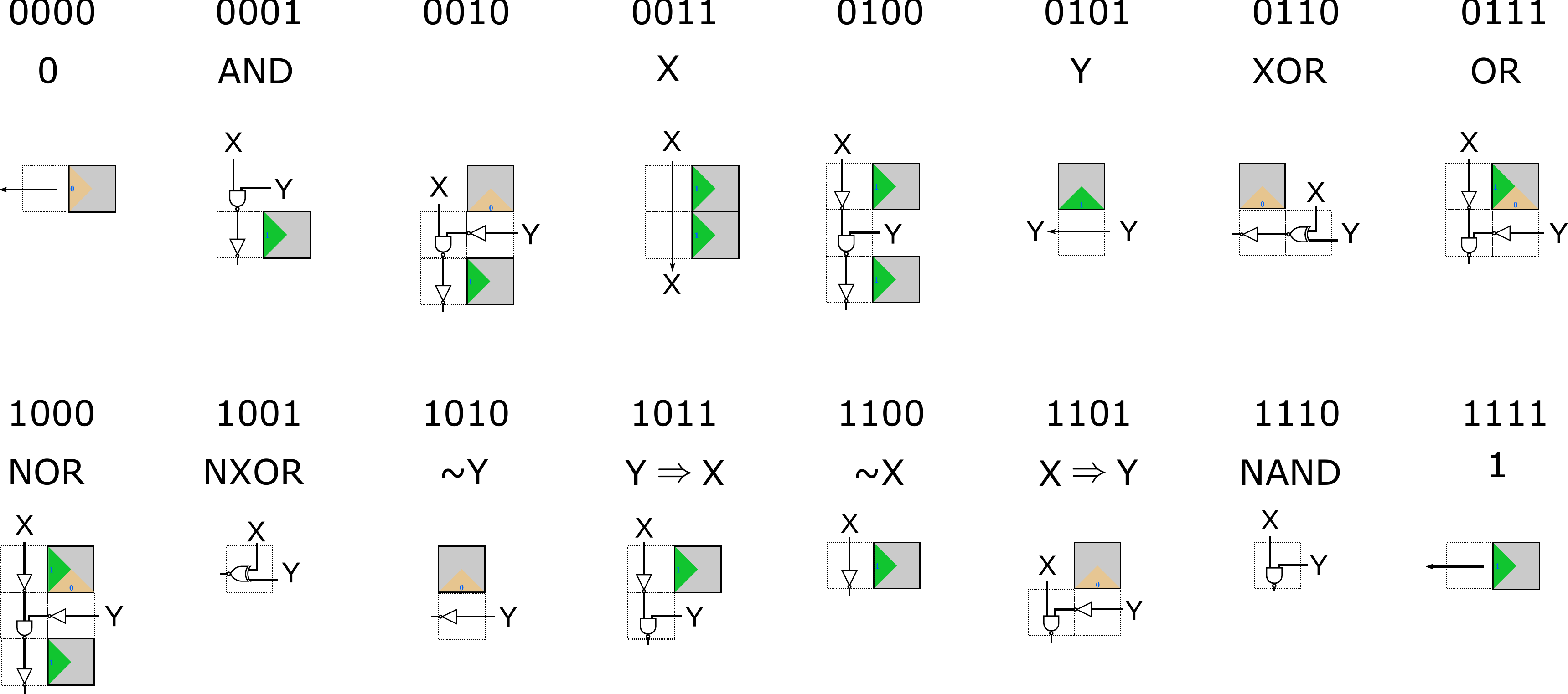}
    \caption{Implementation of all 2-input 1-output Boolean gates using gadgets over the \nn tile set in the \mawatam. 
    The gadgets are ordered with respect to their truth table, which refers to the 4-bit output of the 4 respective inputs 00, 01, 10, 11; i.e.~the canonical truth-table definition of a 2-in 1-out gate (we use the same notation for gates with one (NOT, identity) or zero inputs (constants)). 
    For instance, the truth table $1101$ encodes gate $g$ such that $g(00) = 1$, $g(01) = 1$, $g(10) = 0$ and $g(11) = 1$. The common English name of the gate is also given when there is one. The constant gadgets (0000 and 1111) are used to simulate constant gates (0/1) and circuit input gates $x_i \in \{0,1\}$, and require the presence of an additional glue (not shown) to trigger growth, e.g.~by being placed next to a wire gadget as shown in Figure~\ref{fig:gadgets4}(j4).}\label{fig:all_gates4}
\end{figure}

\input{sixtiles}

\begin{figure}[h]
    \includegraphics[width=\textwidth]{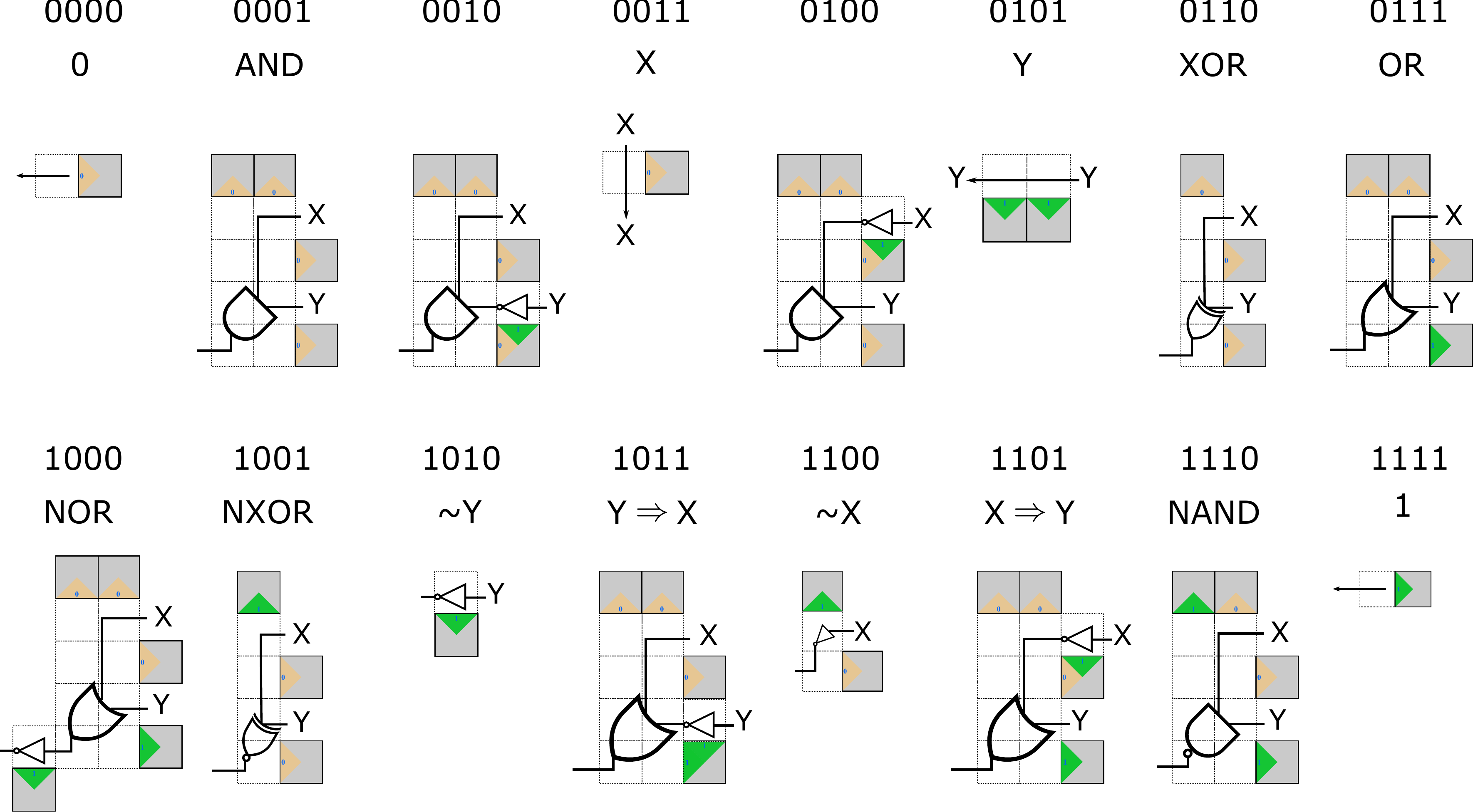}
    \caption{Implementation of all 2-input 1-output Boolean gates using gadgets over the Collatz tile set in the \mawatam. 
    The gadgets are ordered with respect to their truth table which refers to the 4-bit output of the 4 respective inputs 00, 01, 10, 11; i.e.~the canonical truth-table definition of a 2-in 1-out gate (we use the same notation for gates with one (NOT, identity) or zero inputs (constants)).  
    For instance, the truth table $1101$ encodes gate $g$ such that $g(00) = 1$, $g(01) = 1$, $g(10) = 0$ and $g(11) = 1$. 
    The common English name of the gate is also given when there is one.
    The constant gadgets (0000 and 1111) are used to simulate constant gates (0/1) and circuit input gates $x_i \in \{0,1\}$, and require the presence of an additional glue (not shown) to trigger growth, e.g.~by being placed next to a wire gadget as shown in Figure~\ref{fig:gadgets6}(j2).}\label{fig:all_gates6}
\end{figure}

\section*{Acknowledgements}
We thank Trent Rogers, Niall Murphy, Pierre Marcus and Nicolas Schabanel for useful discussions on the \mawatamfull.  

\bibliography{main}

\appendix

\section{Origins of the Collatz tile set}\label{app:iwantmoreCollatz}

\begin{figure}\centering
\includegraphics[width=0.9\textwidth]{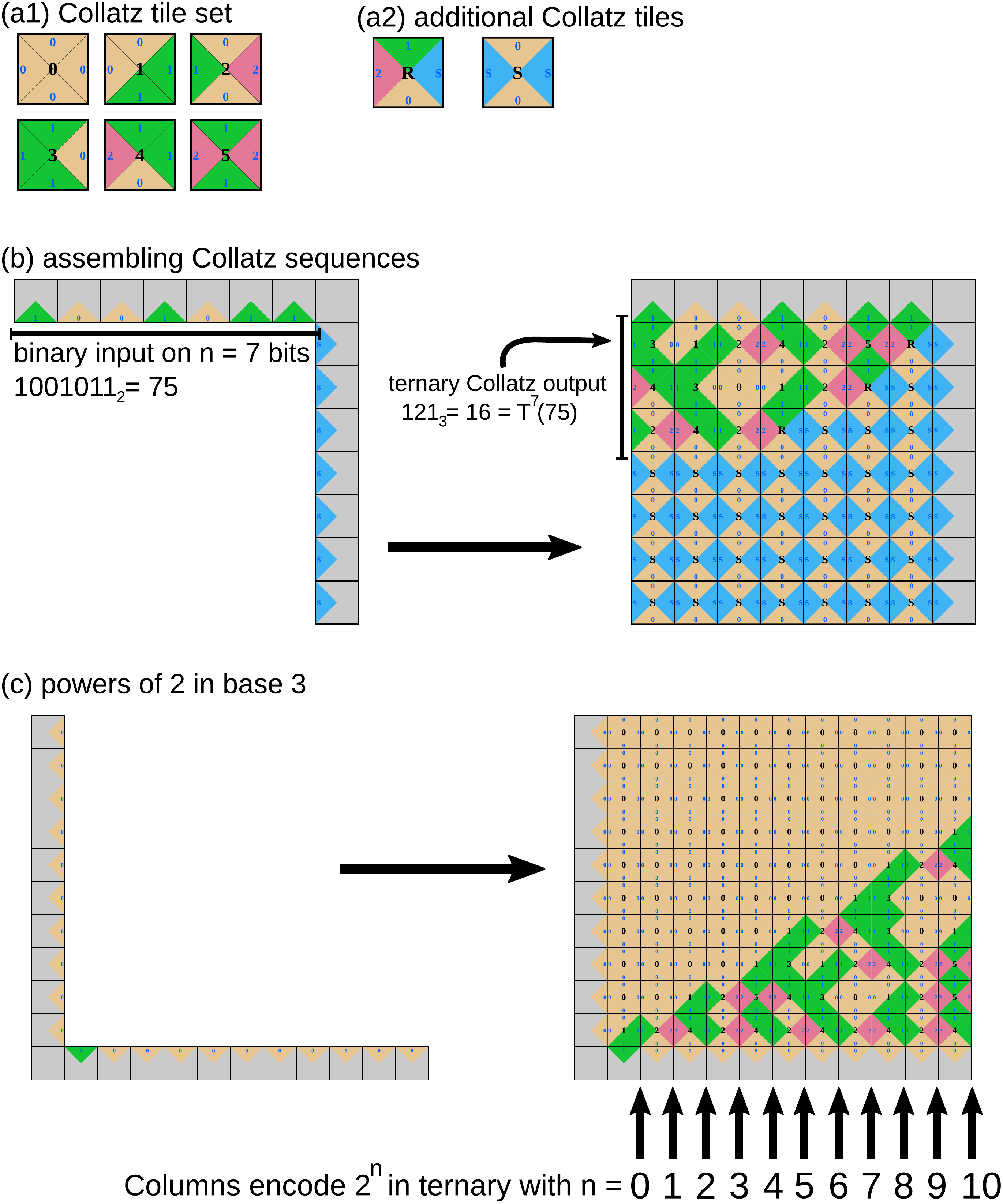}
\caption{The Collatz tile set and its relationship with the Collatz problem and Erd\"{o}s' conjecture. (a1) The Collatz tile set. (a2) Two additional tiles which allow to assemble Collatz trajectories from simple north-east L-shaped seeds. (b) Assembling the first $7$ steps of the Collatz trajectory of $75 = 1001011_2$. The output, $T^7(75)$ can be read in base 3 on the west-most glues of the final assembly (ignoring ``S'' glues). Here, $T^7(1001011_2) = 121_3$ meaning, in base 10, $T^7(75) = 16$. (c) Constructing successive powers of $2$ in base 3: the column marked with arrow number $n$ encodes $2^n$ in base 3.  Erd\"{o}s' conjecture states that, for $n>8$, $2^n$ contains at least one $2$ in base 3 \cite{ErdosPowers2}.}\label{fig:moreCollatz}
\end{figure}

The Collatz problem is a notoriously hard open problem which lies at the intersection
of mathematics and computer science \cite{Lagarias1986,wirsching1998the,survey1}. Formulated in the 30s, the Collatz problem is dauntingly simple to express: consider the Collatz map $T: \mathbb{N} \to \mathbb{N}$ defined by $T(x) = x/2$ if $x$ is even and $T(x) = (3x+1)/2$ if $x$ is odd. The Collatz conjecture states that iterating~$T$, starting from any $n\in \{ 1,2,3,\ldots \}$,  eventually yields $1$.

The Collatz tile set  consists of six tiles, named ``0'' to ``5'', which are depicted in Figure~\ref{fig:moreCollatz}(a1). Vertical glues (north and south) are binary digits (0 and 1) while horizontal glues (east and west) are ternary digits (0, 1 and 2). Each tile is uniquely identified by its north-east corner (pair of glues) or its south-east corner or its south-west corner. Tile names are linked to the tile's glues by the following arithmetical relation: for the tile named~$x$ (with $0 \leq x < 6$) we have:
\begin{equation}
x = 3 N + E = 2 W+S\label{eq:magic}
\end{equation}
where $N, E, W$ and $S$ respectively denote  the values of the North, East, West and South glues. 
 This tile set, among all tile sets which use binary (0, 1) vertical glues and ternary (0, 1, 2) horizontal glues, is the largest tile set for which \eqref{eq:magic} holds, by the following argument. Indeed, \eqref{eq:magic} corresponds to the Euclidean division of $x$ by $3$ and by $2$, meaning that, for a given pair $(N,E)\in\{0,1\}\times\{0,1,2\}$ there is a unique corresponding pair $(S,W)\in\{0,1\}\times\{0,1,2\}$. Since there are 6 different $(N,E)$ pairs we deduce that there are exactly 6 different tiles with binary vertical glues and ternary horizontal glues that satisfy \eqref{eq:magic}. 
 Moreover, analogous tile sets can be generated for any relatively prime $p,q$ (not only $p=2, q=3$), further suggesting its naturalness as an object of study.

{\bf Computing Collatz trajectories with the Collatz tile set plus two more tiles.}
Together with the two extra tiles depicted in Figure~\ref{fig:moreCollatz}(a2), the Collatz tile set is able to assemble Collatz trajectories starting from a straightforward north-east L-shaped seed as depicted in Figure~\ref{fig:moreCollatz}(b). Input to the Collatz iterations are given in binary on the north-most glues (with LSB to the east). 
If the binary input $x$ is of size $n$, we place $n$ ``S'' on the vertical portion of the seed (to the east). The assembly process is directed (i.e.~deterministic in these sense of which tile type is placed where) and, after it is finished, the $n^\text{th}$ Collatz iterate of the binary input $x$, that is $T^{n}(x)$, will be written in ternary along the west-most glues of the assembly (ignoring ``S'' glues). In the example of Figure~\ref{fig:moreCollatz}(b) we read $T^7(1001011_2) = 121_3$ meaning that, in base 10, $T^7(75) = 16$. This phenomenon can be proven using the results of~\cite{Collatz2}, more precisely by identifying the Collatz tile set to the local rule of the CA-like system introduced in \cite{Collatz2} and the two additional tiles to the non-local rule in Figure~(1a)[right] of~\cite{Collatz2}). In practice, the two additional tiles are merely responsible for deleting trailing 0s in binary (which corresponds to the $/2$ part of the Collatz map) while the Collatz tile set does the heavier work of computing $3x+1$ in binary while maintaining a correspondence between base $2$ and base $3$ encodings.

{\bf Predicting patterns produced by the Collatz tile set.}
The computational complexity of predicting what tile will be placed at a given position of the square area defined by the north-east L-shaped seed in Figure~\ref{fig:moreCollatz}(b) is an open question~\cite{Collatz2}. However, if we restrict ourselves to using the Collatz tile set alone, without the two additional tiles, the prediction problem is in NL for each of the three L-shaped seeds: north-east, south-east and south-west (for any length $n \in \Nset$). That is because the relationship between pairs of tile sides, expressed in (1), can be generalised to any rectangular assembly to give a simple arithmetical formula computable in nondeterministic logspace~\cite{Collatz2,hesse2002uniform}  ($3^h N + E = 2^w W+S$, where now $N,E,W,S$ denote binary/ternary numbers written in glue sequences along the respective North, East, West and South sides of a $w \times h$ rectangle). This fact means that, assuming a widely-believed conjecture in complexity theory (namely, NL $\neq$ P), it is not possible to simulate arbitrary, polynomial size, Boolean circuits using the 6-tile Collatz tile set with those simple L-shaped connected seeds (within area polynomial in circuit size).

Although rectangular assemblies made with the Collatz tile set are simple to predict, they also relate to hard open questions in number theory. Notably to the following conjecture by Erd\"{o}s \cite{ErdosPowers2}: For all $n > 8$, there is at least one digit $2$ in the ternary representation of~$2^n$.
Indeed, starting from the straightforward south-west L-shaped seed of Figure~\ref{fig:moreCollatz}(c), consisting of $m$ vertical $0$s and a horizontal $1$ followed by $m-1$ horizontal $0$s, an induction proves that consecutive columns of the assembly will encode successive powers of two in ternary. For instance, on the first 4 columns pointed by an arrow in Figure~\ref{fig:moreCollatz}(c) we can successively read: ``1'', ``2'', ``11'', ``22'' which are the ternary encodings of 1, 2, 4 and 8,  the four first powers of 2. Erd\"{o}s conjecture then becomes: any column to the east of the $10^\text{th}$ column of the assembly (counting from the easternmost input column of vertical $0$s), will contain a glue ``2'' (in red). This problem can be seen as a potentially simpler conjecture than the Collatz conjecture~\cite{LagariasPowers2}.

\end{document}

%% file: sixtiles.tex

\section{Six tiles: the Collatz tileset}\label{sec:six tiles}

In this section, we illustrate efficient Boolean circuit simulation in the \datam with 
the Collatz tile set which consists of of 6 tile types and 3 glues and is shown in Figure~\ref{fig:tilesets}(b).

On the one hand, the \nn tile set was explicitly designed to compute, via the placement of a single tile, the universal NAND function. From there it was augmented (with bits on the west sides) that facilitate simulation of circuit wiring, and efficient simulation (few tiles) of non-NAND gates.   
On the other hand, the Collatz tile set came about from studies on the Collatz problem. 
Specifically, glue patterns in some tiled regions  
(e.g.~rectangles) relate to notoriously hard mathematical problems such as the Collatz conjecture \cite{Collatz2} or an open problem of  Erd\"{o}s'~\cite{ErdosPowers2,LagariasPowers2}: 
Is it the case that for all $n>8$ there is at least one $2$ in the ternary representation of $2^n$?
For more details see Appendix~\ref{app:iwantmoreCollatz}.
We noticed that this pattern complexity could be leveraged, with the aid of computer search\footnote{Computer search was performed through the \mawatam simulator: \url{https://github.com/tcosmo/mawatam}},
 to build gadgets for computation in the \mawatam (Figure~\ref{fig:gadgets6}).

\thmCollatz* 
\begin{figure}
\centering
  \includegraphics[width=0.95\textwidth]{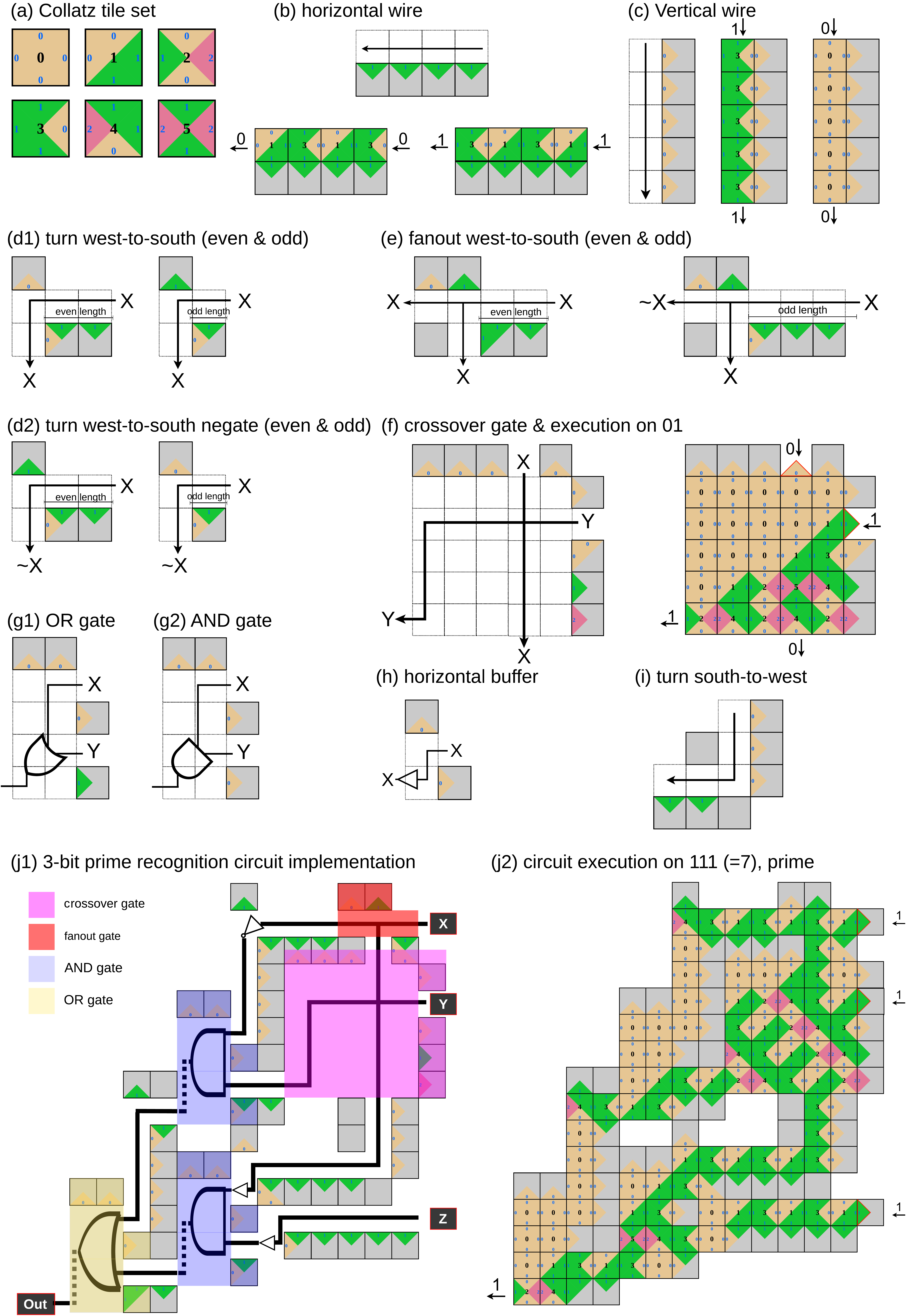}
  \caption{Circuit-simulating gadgets for the Collatz tile set. Growth proceeds to the west and south exclusively. (a) the Collatz tile set. Seed structures to implement (b) horizontal west-growing and (c) vertical south-growing wires. Horizontal wires are of even length. When turning to the south the appropriate turn can be used to transmit the signal (d1) or its negation (d2). (e) Fanout gadgets depending on the parity of the incoming horizontal wire, if the length is odd, the gadget also negates the west-going signal. (f) The smallest crossover gate found by computer search. (g) Common Boolean gates, also found by computer search. (h) The buffer gadget is used to change the parity of an horizontal wire. (i) Turn south-to-west. (j1) Collatz-tileset implementation of the 3-bit prime recognition circuit and (j2) execution of the circuit on $7_{10} = 111_2$ which is prime.}\label{fig:gadgets6}
\end{figure}
\begin{proof}
    {\bf Tiles simulating wires and gates.} 
    We will show that the gadgets in Figure~\ref{fig:gadgets6} can be used to build mazes that simulate arbitrary Boolean circuits and 
    that the growth triggered by the placement of input tiles is directed, 
    which in turn implies that the correct bit is output by the simulation of $c$ on some binary input word $x$.

    Figure~\ref{fig:gadgets6}(b,c) details how the Collatz tile set simulates horizontal and vertical wires. 
    Horizontal tile-wires have a parity constraint: in a horizontal wire carrying the bit $x\in\{0,1\}$, 
    every second tile correctly advertises $x$ to the west, and every other tile advertises its negation~$\simnot x$. 
    To handle this, there are two west-to-south turns, one for turning from even length, and one for turning from odd length, horizontal wires Figure~\ref{fig:gadgets6}(d1). Only west-to-south fanout is used in the constructions with this tileset, Figure~\ref{fig:gadgets6}(e). This fanout gate comes in two variants whether it is applied at an even or an odd horizontal wire position. If the gadget is applied at an odd wire position, it has the particularity of negating the output west-going signal. 

    Negating a signal (either to correct a horizontal parity effect, or to simulate a NOT gate) can be achieved in several ways. If the signal ever turns south, this can easily be done thanks to Figure~\ref{fig:gadgets6}(d2) which implements both a turn and a negation at the same time.
    If the signal never turns south, the programmer can use an odd-length horizontal wire which implements a negation. If using an odd-length horizontal wire is not possible given the constraints on circuit layout, the programmer can use the horizontal buffer gadget Figure~\ref{fig:gadgets6}(h) which has the effect of copying the incoming signal to the next immediate column to the west which inverts the parity constraint of the horizontal wire and allows it to reproduce the behavior of an odd-length horizontal wire. This method is used in Figure~\ref{fig:gadgets6}(j1), for instance on the horizontal wire which connects the input Z to its target AND gate. 

    Glue labelled polyominos, or seed structures, for south-to-west turns is shown in Figure~\ref{fig:gadgets6}(i). Notably, a growth stopper ($1\times 1$ polyomino, with four null glues) is used to prevent spurious growth that would happen in the north-west direction  otherwise.

    A crossover gadget seed structure is given in Figure~\ref{fig:gadgets6}(f), it was the smallest found by computer search and it costs 33 tiles. The gate preserves the horizontal alignment of the incoming northern bit: it exits at the south of the gate at the same $x$-position that it entered. However, the incoming eastern bit is deviated three units to the south.

    Seed (polyomino) structures that simulate Boolean (compute) gates
     are rectangular and were found by computer search using the input convention that signals come from the east and, if there are two of them the inputs should be one vertical block apart\footnote{Using computer search, we were able to find rectangular seed structures of Boolean gates corresponding to all the input conventions that we experimented with. This leads us to believe that the ability of the Collatz tileset to simulate Boolean gates is not tied to a particular input convention.}. Figure~\ref{fig:gadgets6}(g1,g2) gives the seed structure of an OR gate and an AND gate. For completeness, Figure~\ref{fig:all_gates6} gives the implementation of all Boolean gates, the biggest of them is NOR with a cost of 14 tiles. This gives the tiles bounds per gate and crossover in the theorem statement. Remarkably, seed structures for AND, OR, NAND, NOR are very similar in the sense that they differ by at most 2 glues.

We claim that each gadget in Figure~\ref{fig:gadgets6} and Figure~\ref{fig:all_gates6} is directed, meaning that after input glues are supplied to the gadget then for each dotted region in the gadget there is exactly one tile type that can be placed. This can be seen by noting that (i) all gadgets use either North and East sides to attach or South and East sides to attach (South and East attachments are only used for horizontal wires and turn south-to-west gadgets, Figure~\ref{fig:gadgets6}(b,i)), 
(ii) North and East attachments cannot compete with South and East attachments because all signals travel in the south-west direction and South and East constraints are never given directly by the seed but occur after tiles attach, and 
(iii) the Collatz tile set is deterministic on North and East sides and South and East sides.

{\bf Laying the circuit out on a grid.}
We use the same circuit layout technique given in the the proof of Theorem~\ref{thm:NANDNXOR}. 

{\bf Computation.}
Similarly to the proof of Theorem~\ref{thm:NANDNXOR}, throughout the entire assembly process, because of the directedness of all the gadgets that we use, at each position there is exactly one tile type that can be placed. Thus one final assembly is produced, that encodes an execution of the circuit, and in particular outputs the same bit as the $n$-bit circuit $c$ on any input word $x \in \{ 0,1\}^n$.
\end{proof}

\begin{example}\normalfont
The 3-bit prime recognition circuit in Figure~\ref{fig:gadgets4}(j1,j2) is implemented using the Collatz tile set in  Figure~\ref{fig:gadgets6}(j1,j2).
\end{example}

%% file: main.bbl
\begin{thebibliography}{10}

\bibitem{allender2009planar}
Eric Allender, David A~Mix Barrington, Tanmoy Chakraborty, Samir Datta, and
  Sambuddha Roy.
\newblock Planar and grid graph reachability problems.
\newblock {\em Theory of Computing Systems}, 45(4):675--723, 2009.

\bibitem{surface_swap_Brailovskaya_2019}
Tatiana Brailovskaya, Gokul Gowri, Sean Yu, and Erik Winfree.
\newblock Reversible computation using swap reactions on a surface.
\newblock In Chris Thachuk and Yan Liu, editors, {\em DNA Computing and
  Molecular Programming}, pages 174--196, Cham, 2019. Springer International
  Publishing.

\bibitem{localised_dna_hairpin_chain_Bui}
Hieu Bui, Vincent Miao, Sudhanshu Garg, Reem Mokhtar, Tianqi Song, and John
  Reif.
\newblock Design and analysis of localized {DNA} hybridization chain reactions.
\newblock {\em Small}, 13(12):1602983, 2017.

\bibitem{localised_dna_origami_circuit}
Hieu Bui, Shalin Shah, Reem Mokhtar, Tianqi Song, Sudhanshu Garg, and John
  Reif.
\newblock Localized {DNA} hybridization chain reactions on {DNA} origami.
\newblock {\em ACS Nano}, 12(2):1146--1155, Feb 2018.

\bibitem{cantu2021covert}
Angel~A Cantu, Austin Luchsinger, Robert Schweller, and Tim Wylie.
\newblock Covert computation in self-assembled circuits.
\newblock {\em Algorithmica}, 83(2):531--552, 2021.
\newblock arXiv preprint: \href{https://arxiv.org/abs/1908.06068}{\tt
  arXiv:1908.06068}.

\bibitem{pmid28977499}
A.~R. Chandrasekaran, O.~Levchenko, D.~S. Patel, M.~MacIsaac, and K.~Halvorsen.
\newblock {{A}ddressable configurations of {D}{N}{A} nanostructures for
  rewritable memory}.
\newblock {\em Nucleic Acids Res}, 45(19):11459--11465, Nov 2017.

\bibitem{Chao2019}
Jie Chao, Jianbang Wang, Fei Wang, Xiangyuan Ouyang, Enzo Kopperger, Huajie
  Liu, Qian Li, Jiye Shi, Lihua Wang, Jun Hu, Lianhui Wang, Wei Huang,
  Friedrich~C. Simmel, and Chunhai Fan.
\newblock Solving mazes with single-molecule {DNA} navigators.
\newblock {\em Nature Materials}, 18(3):273--279, Mar 2019.

\bibitem{Chatterjee2017}
Gourab Chatterjee, Neil Dalchau, Richard~A. Muscat, Andrew Phillips, and Georg
  Seelig.
\newblock A spatially localized architecture for fast and modular {DNA}
  computing.
\newblock {\em Nature Nanotechnology}, 12(9):920--927, Sep 2017.

\bibitem{graphOnAGrid}
Marek Chrobak and Thomas~H Payne.
\newblock A linear-time algorithm for drawing a planar graph on a grid.
\newblock {\em Information Processing Letters}, 54(4):241--246, 1995.

\bibitem{surface_clamons2020}
Samuel Clamons, Lulu Qian, and Erik Winfree.
\newblock Programming and simulating chemical reaction networks on a surface.
\newblock {\em Journal of The Royal Society Interface}, 17(166):20190790, 2020.

\bibitem{Cook}
Matthew Cook.
\newblock Universality in elementary cellular automata.
\newblock {\em Complex Systems}, 15, 2004.

\bibitem{dalchau2015probabilistic}
Neil Dalchau, Harish Chandran, Nikhil Gopalkrishnan, Andrew Phillips, and John
  Reif.
\newblock Probabilistic analysis of localized {DNA} hybridization circuits.
\newblock {\em ACS synthetic biology}, 4(8):898--913, 2015.

\bibitem{Dannenberg}
Frits Dannenberg, Marta Kwiatkowska, Chris Thachuk, and Andrew~J. Turberfield.
\newblock Dna walker circuits: Computational potential, design, and
  verification.
\newblock In David Soloveichik and Bernard Yurke, editors, {\em DNA Computing
  and Molecular Programming}, pages 31--45, Cham, 2013. Springer International
  Publishing.

\bibitem{OneTile}
Erik~D. Demaine, Martin~L. Demaine, S{\'a}ndor~P. Fekete, Matthew~J. Patitz,
  Robert~T. Schweller, Andrew Winslow, and Damien Woods.
\newblock One tile to rule them all: Simulating any tile assembly system with a
  single universal tile.
\newblock In {\em ICALP: Proceedings of the 41st International Colloquium on
  Automata, Languages, and Programming}, volume 8572 of {\em LNCS}, pages
  368--379. Springer, 2014.
\newblock Arxiv preprint: \href{http://arxiv.org/abs/1212.4756}{\tt
  arXiv:1212.4756}.

\bibitem{2HAMIU}
Erik~D. Demaine, Matthew~J. Patitz, Trent~A. Rogers, Robert~T. Schweller,
  Scott~M. Summers, and Damien Woods.
\newblock The two-handed tile assembly model is not intrinsically universal.
\newblock In {\em ICALP: Proceedings of the 40th International Colloquium on
  Automata, Languages, and Programming}, volume 7965 of {\em LNCS}, pages
  400--412. Springer, July 2013.
\newblock Arxiv preprint: \href{http://arxiv.org/abs/1306.6710}{\tt
  arXiv:1306.6710}.

\bibitem{DotCACM}
David Doty.
\newblock Theory of algorithmic self-assembly.
\newblock {\em Communications of the ACM}, 55(12):78--88, 2012.

\bibitem{IUSA}
David Doty, Jack~H. Lutz, Matthew~J. Patitz, Robert~T. Schweller, Scott~M.
  Summers, and Damien Woods.
\newblock The tile assembly model is intrinsically universal.
\newblock In {\em FOCS: Proceedings of the 53rd Annual IEEE Symposium on
  Foundations of Computer Science}, pages 439--446. IEEE, October 2012.
\newblock Arxiv preprint: \href{http://arxiv.org/abs/1111.3097}{\tt
  arXiv:1111.3097}.

\bibitem{ErdosPowers2}
Paul Erdös.
\newblock Some unconventional problems in number theory.
\newblock {\em Mathematics Magazine}, 52(2):67--70, 1979.
\newblock URL: \url{https://doi.org/10.1080/0025570X.1979.11976756}.

\bibitem{evans2014crystals}
Constantine Evans.
\newblock {\em Crystals that count! {P}hysical principles and experimental
  investigations of {{DNA}} tile self-assembly}.
\newblock PhD thesis, Caltech, 2014.

\bibitem{gu2010proximity}
Hongzhou Gu, Jie Chao, Shou-Jun Xiao, and Nadrian~C Seeman.
\newblock A proximity-based programmable {DNA} nanoscale assembly line.
\newblock {\em Nature}, 465(7295):202--205, 2010.

\bibitem{hesse2002uniform}
William Hesse, Eric Allender, and David A~Mix Barrington.
\newblock Uniform constant-depth threshold circuits for division and iterated
  multiplication.
\newblock {\em Journal of Computer and System Sciences}, 65(4):695--716, 2002.

\bibitem{Imm1998}
Neil Immerman.
\newblock {\em Descriptive Complexity}.
\newblock Springer, 1999.

\bibitem{Lagarias1986}
Jeffrey~C. Lagarias.
\newblock The $3x + 1$ problem and its generalizations.
\newblock {\em The American Mathematical Monthly}, 92(1):3--23, 1985.
\newblock URL: \url{http://www.jstor.org/stable/2322189}.

\bibitem{survey1}
Jeffrey~C. Lagarias.
\newblock The $3x+1$ problem: An annotated bibliography (1963--1999) (sorted by
  author), 2003.
\newblock \href {http://arxiv.org/abs/arXiv:math/0309224}
  {\path{arXiv:arXiv:math/0309224}}.

\bibitem{LagariasPowers2}
Jeffrey~C. Lagarias.
\newblock Ternary expansions of powers of 2.
\newblock {\em Journal of the London Mathematical Society}, 79(3):562--588,
  2009.
\newblock \href {https://doi.org/https://doi.org/10.1112/jlms/jdn080}
  {\path{doi:https://doi.org/10.1112/jlms/jdn080}}.

\bibitem{lakin2014abstract}
Matthew~R Lakin, Rasmus Petersen, Kathryn~E Gray, and Andrew Phillips.
\newblock Abstract modelling of tethered {DNA} circuits.
\newblock In {\em International Workshop on {DNA}-Based Computers}, pages
  132--147. Springer, 2014.

\bibitem{liu2011crystalline}
Wenyan Liu, Hong Zhong, Risheng Wang, and Nadrian~C Seeman.
\newblock Crystalline two-dimensional {DNA}-origami arrays.
\newblock {\em Angewandte Chemie International Edition}, 50(1):264--267, 2011.

\bibitem{moore2011nature}
Cristopher Moore and Stephan Mertens.
\newblock {\em The nature of computation}.
\newblock Oxford University Press, 2011.

\bibitem{murphy2012}
Niall Murphy and Damien Woods.
\newblock {AND} and/or {OR}: Uniform polynomial-size circuits.
\newblock In {\em MCU 2013: Machines, Computations and Universality. Electronic
  Proceedings in Theoretical Computer Science (EPTCS)}, volume 128, pages
  150--166, 2012.
\newblock \href{https://arxiv.org/abs/1212.3282v2}{\tt arXiv:1212.3282v2}.

\bibitem{neary2006p}
Turlough Neary and Damien Woods.
\newblock P-completeness of cellular automaton {R}ule 110.
\newblock In {\em ICALP: International Colloquium on Automata, Languages, and
  Programming}, volume 4051, part 1 of {\em LNCS}, pages 132--143. Springer,
  2006.

\bibitem{neary2009four}
Turlough Neary and Damien Woods.
\newblock Four small universal {T}uring machines.
\newblock {\em Fundamenta Informaticae}, 91(1):123--144, 2009.

\bibitem{NearyWoodsWeakly}
Turlough Neary and Damien Woods.
\newblock Small weakly universal {T}uring machines.
\newblock In {\em International Symposium on Fundamentals of Computation
  Theory}, pages 262--273. Springer, 2009.

\bibitem{walker_omabegho2009}
Tosan Omabegho, Ruojie Sha, and Nadrian~C Seeman.
\newblock A bipedal dna brownian motor with coordinated legs.
\newblock {\em Science}, 324(5923):67--71, 2009.

\bibitem{pardatscher2016dna}
G{\"u}nther Pardatscher, Dan Bracha, Shirley~S Daube, Ohad Vonshak, Friedrich~C
  Simmel, and Roy~H Bar-Ziv.
\newblock {DNA} condensation in one dimension.
\newblock {\em Nature nanotechnology}, 11(12):1076--1081, 2016.

\bibitem{PatitzSurveyJournal}
Matthew~J. Patitz.
\newblock An introduction to tile-based self-assembly and a survey of recent
  results.
\newblock {\em Natural Computing}, 13(2):195--224, 2014.

\bibitem{surface_qian2014}
Lulu Qian and Erik Winfree.
\newblock Parallel and scalable computation and spatial dynamics with dna-based
  chemical reaction networks on a surface.
\newblock In Satoshi Murata and Satoshi Kobayashi, editors, {\em DNA Computing
  and Molecular Programming}, pages 114--131, Cham, 2014. Springer
  International Publishing.

\bibitem{walker_computing_REIF20091428}
John~H. Reif and Sudheer Sahu.
\newblock Autonomous programmable {DNA} nanorobotic devices using dnazymes.
\newblock {\em Theoretical Computer Science}, 410(15):1428--1439, 2009.
\newblock Aspects of Molecular Self-Assembly.

\bibitem{Rogozhin96}
Yuri Rogozhin.
\newblock Small universal turing machines.
\newblock {\em Theoretical Computer Science}, 168(2):215–240, 1996.

\bibitem{rothemund2000program}
Paul W~K Rothemund and Erik Winfree.
\newblock The program-size complexity of self-assembled squares.
\newblock In {\em STOC: Proceedings of the thirty-second annual ACM symposium
  on Theory of computing}, pages 459--468. ACM, 2000.

\bibitem{rothemund2006folding}
Paul~WK Rothemund.
\newblock Folding {DNA} to create nanoscale shapes and patterns.
\newblock {\em Nature}, 440(7082):297--302, 2006.

\bibitem{walker_Sahu2008}
Sudheer Sahu, Thomas~H LaBean, and John~H Reif.
\newblock A {DNA} nanotransport device powered by polymerase phi29.
\newblock {\em Nano letters}, 8(11):3870—3878, November 2008.

\bibitem{walker_sherman2004}
William~B Sherman and Nadrian~C Seeman.
\newblock A precisely controlled {DNA} biped walking device.
\newblock {\em Nano letters}, 4(7):1203--1207, 2004.

\bibitem{SolWin07}
David Soloveichik and Erik Winfree.
\newblock Complexity of self-assembled shapes.
\newblock {\em SIAM Journal on Computing}, 36(6):1544--1569, 2007.

\bibitem{SoloveichikWinfree}
David Soloveichik and Erik Winfree.
\newblock Complexity of self-assembled shapes.
\newblock {\em SIAM Journal on Computing}, 36(6):1544--1569, 2007.

\bibitem{dna_based_circuit_cancer_membrane_Song2019}
Tianqi Song, Shalin Shah, Hieu Bui, Sudhanshu Garg, Abeer Eshra, Daniel Fu,
  Ming Yang, Reem Mokhtar, and John Reif.
\newblock Programming dna-based biomolecular reaction networks on cancer cell
  membranes.
\newblock {\em Journal of the American Chemical Society}, 141(42):16539--16543,
  Oct 2019.

\bibitem{darkoMazeWalker}
Darko Stefanovic.
\newblock Maze exploration with molecular-scale walkers.
\newblock In Adrian-Horia Dediu, Carlos Mart{\'i}n-Vide, and Bianca Truthe,
  editors, {\em Theory and Practice of Natural Computing}, pages 216--226,
  Berlin, Heidelberg, 2012. Springer Berlin Heidelberg.

\bibitem{Collatz2}
Tristan St{\'e}rin and Damien Woods.
\newblock The {C}ollatz process embeds a base conversion algorithm.
\newblock In Sylvain Schmitz and Igor Potapov, editors, {\em RP2020: 14th
  International Conference on Reachability Problems}, volume 12448 of {\em
  LNCS}, pages 131--147. Springer, 2020.
\newblock \href{https://arxiv.org/abs/2007.06979}{\tt arXiv:2007.06979}
  [cs.DM].

\bibitem{walker_cargoSorting2017}
Anupama~J. Thubagere, Wei Li, Robert~F. Johnson, Zibo Chen, Shayan Doroudi,
  Yae~Lim Lee, Gregory Izatt, Sarah Wittman, Niranjan Srinivas, Damien Woods,
  Erik Winfree, and Lulu Qian.
\newblock A cargo-sorting {DNA} robot.
\newblock {\em Science}, 357(6356), 2017.

\bibitem{tikhomirov2017fractal}
Grigory Tikhomirov, Philip Petersen, and Lulu Qian.
\newblock Fractal assembly of micrometre-scale {DNA} origami arrays with
  arbitrary patterns.
\newblock {\em Nature}, 552(7683):67--71, 2017.

\bibitem{tikhomirov2017programmable}
Grigory Tikhomirov, Philip Petersen, and Lulu Qian.
\newblock Programmable disorder in random {DNA} tilings.
\newblock {\em Nature nanotechnology}, 12(3):251, 2017.

\bibitem{SST2D}
Bryan Wei, Mingjie Dai, and Peng Yin.
\newblock Complex shapes self-assembled from single-stranded {DNA} tiles.
\newblock {\em Nature}, 485(7400):623--626, 2012.

\bibitem{Wickham2012}
Shelley F.~J. Wickham, Jonathan Bath, Yousuke Katsuda, Masayuki Endo, Kumi
  Hidaka, Hiroshi Sugiyama, and Andrew~J. Turberfield.
\newblock A {DNA}-based molecular motor that can navigate a network of tracks.
\newblock {\em Nature Nanotechnology}, 7(3):169--173, Mar 2012.
\newblock \href {https://doi.org/10.1038/nnano.2011.253}
  {\path{doi:10.1038/nnano.2011.253}}.

\bibitem{Winf98}
Erik Winfree.
\newblock {\em Algorithmic Self-Assembly of {DNA}}.
\newblock PhD thesis, California Institute of Technology, 1998.

\bibitem{erikNed1998}
Erik Winfree, Furong Liu, Lisa~A Wenzler, and Nadrian~C Seeman.
\newblock Design and self-assembly of two-dimensional dna crystals.
\newblock {\em Nature}, 394(6693):539--544, 1998.

\bibitem{wirsching1998the}
Günther~J. Wirsching.
\newblock {\em The dynamical system generated by the 3n + 1 function}.
\newblock Springer, Berlin New York, 1998.

\bibitem{woo2011programmable}
Sungwook Woo and Paul~WK Rothemund.
\newblock Programmable molecular recognition based on the geometry of {DNA}
  nanostructures.
\newblock {\em Nature chemistry}, 3(8):620, 2011.

\bibitem{IUsurvey}
Damien Woods.
\newblock Intrinsic universality and the computational power of self-assembly.
\newblock {\em Philosophical Transactions of the Royal Society A: Mathematical,
  Physical and Engineering Sciences}, 373(2046):20140214, 2015.

\bibitem{IBCs}
Damien Woods, David Doty, Cameron Myhrvold, Joy Hui, Felix Zhou, Peng Yin, and
  Erik Winfree.
\newblock Diverse and robust molecular algorithms using reprogrammable {DNA}
  self-assembly.
\newblock {\em Nature}, 567(7748):366--372, 2019.

\bibitem{WoodsNearySurvey}
Damien Woods and Turlough Neary.
\newblock The complexity of small universal {T}uring machines: A survey.
\newblock {\em Theoretical Computer Science}, 410(4-5):443--450, 2009.

\bibitem{Yan8103}
Hao Yan, Thomas~H. LaBean, Liping Feng, and John~H. Reif.
\newblock Directed nucleation assembly of {DNA} tile complexes for
  barcode-patterned lattices.
\newblock {\em Proceedings of the National Academy of Sciences},
  100(14):8103--8108, 2003.

\bibitem{walker_yin2004}
P.~Yin, H.~Yan, X.~G. Daniell, A.~J. Turberfield, and J.~H. Reif.
\newblock {{A} unidirectional {DNA} walker that moves autonomously along a
  track}.
\newblock {\em Angewandte Chemie}, 43(37):4906--4911, Sep 2004.

\end{thebibliography}
